\documentclass[11pt]{article}

\usepackage[normalem]{ulem}

\usepackage{amsfonts,amsmath,amsthm,times,fullpage,amssymb,bbm}

\usepackage{hyperref}

\hypersetup{colorlinks=false}

\usepackage[dvipsnames]{xcolor}

\usepackage[noend]{algpseudocode} 
\usepackage[nothing]{algorithm}  
\usepackage{caption
}
 
\newtheorem{theorem}{Theorem} 
\newtheorem{lemma}{Lemma}
\newtheorem{definition}{Definition} 
 
\newtheorem{remark}{Remark}

\newtheorem*{backw}{Congestion}

\newtheorem*{neigh}{Neighborhood}
\newtheorem*{causality}{Causality}

\newcommand{\principal}{\mathrm{pr}}
\newcommand{\noise}{\mathrm{ns}}

\newcommand{\co}{\mathrm{b}}
 
\def\ex{\mathbb E}
\newcommand{\Parent}{\mathrm{Potential}^{-}}

\newcommand{\pot}{\mathrm{Potential}}
\newcommand{\Pot}{\mathrm{Potential}}
\newcommand{\Cong}{\mathrm{Congestion}}

\begin{document}

\title{Stochastic Control via Entropy Compression}

\author{
Dimitris Achlioptas
\thanks{Research supported by NSF grant CCF-1514128.}\\ 
UC Santa Cruz\\
{\small \texttt{optas@cs.ucsc.edu}}
\\
\and 
Fotis Iliopoulos\thanks{ 
Research supported by NSF grant CCF-1514434. Part of of this work was done while at Adobe Research.} \\ 
UC Berkeley\\
{\small \texttt{fotis.iliopoulos@berkeley.edu}}
\and
Nikos Vlassis\\
Adobe Research\\
{\small \texttt{vlassis@adobe.com}}
}

\date{\empty}

\maketitle

\begin{abstract}
We consider an agent trying to bring a system to an acceptable state by repeated \mbox{probabilistic action.} Several recent works on algorithmizations of the Lov\'{a}sz Local Lemma (LLL) can be seen as establishing sufficient conditions for the agent to succeed. Here we study whether such stochastic control is also possible in a noisy environment, where both the process of state-observation and the process of state-evolution are subject to adversarial perturbation (noise). The introduction of noise causes the tools developed for LLL algorithmization to break down since the key LLL ingredient, the sparsity of the causality (dependence) relationship, no longer holds. To overcome this challenge we develop a new analysis where entropy plays a central role, both to measure the rate at which progress towards an acceptable state is made and the rate at which noise undoes this progress. The end result is a sufficient condition that allows a smooth tradeoff between the intensity of the noise and the amenability of the system, recovering an asymmetric LLL condition in the noiseless case.
\end{abstract}

\newpage 

\section{Introduction}\label{setting}

Consider a system with a large state space $\Omega$, hidden from view inside a box. On the outside of the box there are lightbulbs and buttons. Each lightbulb corresponds to a set $f_i \subseteq \Omega$ and is lit whenever the current state of the system is in $f_i$. We think of each set $f_i$ as containing all states sharing some negative feature $i \in [m]$ and refer to each such set as a \emph{flaw}, letting $F = \{f_1, f_2, \ldots, f_m\}$. 
For example, if the system corresponds to a 
graph $G$ with $n$ vertices each of which can take one of $q$ colors, then $\Omega = [q]^n$, and we can define for each edge $e_i$ of $G$ the flaw $f_i$ to contain all assignments of colors to the vertices of $G$ that assign the same color to the endpoints of $e_i$. 
Following linguistic convention, instead of mathematical, we will say that flaw $f$ is present in state $\sigma$ whenever $f \ni \sigma$ and that state $\sigma$ is \emph{flawless} if no flaw is present in $\sigma$. The buttons correspond to actions, i.e., to mechanisms for state evolution. Specifically, taking action $a$ while in state $\sigma$ moves the system to a new state $\tau$, selected from a probability distribution that depends on both $\sigma$ and $a$. 

Outside the box, an agent called the \emph{controller} observes the lightbulbs and pushes buttons, in an effort to bring the system to a flawless state. Specifically, if $O(\sigma) \in \{0,1\}^m$ denotes the lightbulb bitvector, with 1 corresponding to lit, the controller repeatedly applies a function $P$, called a \emph{policy}, that maps $O(\sigma)$ to a distribution over actions. Thus, overall, state evolution proceeds as follows: if the current (hidden) state is $\sigma \in\Omega$, the controller observes $O(\sigma)$ and samples an action from $P(O(\sigma))$; after she takes the chosen action, the system, internally and probabilistically, moves to a new (hidden) state $\tau$, selected from a distribution that depends on both $\sigma$ and the action taken. 

Our work begins with the observation that several recent results~\cite{moser,MT,szege_meet, SrinivasanPerm, AI, JACM, HV, Harmonic, Commu} on LLL algorithmization can be seen as giving sufficient conditions for a controller as above to be able to bring the system to a flawless state quickly, with high probability. Motivated by this viewpoint we ask if conditions for LLL algorithmizations can be seen as \emph{stability criteria} and  give results for more general settings, e.g., Partially Observable Markov Decision Processes (POMDPs).  Given the capacity of LLL algorithmization arguments to establish convergence in highly non-convex domains, a major pain point in control theory, we believe that bringing such arguments to stochastic control is a first step in a fruitful direction.  In order to move in that direction we generalize the setting described so far in two ways:
\begin{itemize}
\item
The mapping $O$ from states to observations is \emph{stochastic}: the lightbulbs are unreliable, exhibiting both false-positives and false-negatives. 
\item 
Both the environment surrounding the system and the implementation of actions are \emph{noisy}: the controller is not the only agent affecting state evolution and flaws may be introduced into the state for reasons unrelated to her actions, even spontaneously. 
\end{itemize}

The question, naturally, is whether sufficient conditions for quick convergence to flawless states can still be established in this new setting. We answer the question affirmatively and show, in a precise mathematical sense, that the less internal conflict there is in the system, the more noise the controller can tolerate. In order to prove this we require the controller to be \emph{focused} and to \emph{prioritize}. That is, we will assume that the flaws are ordered by priority according to an arbitrary but fixed permutation $\pi$ of $F$, and we will ascribe the action taken by the controller in each step to the present flaw (focus) of highest priority (prioritization). The analysis will then take into account both how good the actions are at ridding the state of that flaw and how damaging they are in terms of introducing new flaws. In particular, with this attribution mechanism in place, and similarly to LLL algorithmization arguments, we will say that flaw $f_i$ can cause flaw $f_j$ if there exists a state transition with non-zero probability under the policy, from a state in which $f_i$ is the highest priority flaw and $f_j$ is absent to a state in which $f_j$ is present. 

The main challenge we face is that in the presence of noise the causality relationship becomes dense. To overcome this we develop a new analysis in which causality is not a binary relationship, but one weighted by the \emph{frequency} of interactions. In particular, our condition guaranteeing that the controller will succeed within a reasonable amount of time allows the causality graph to become arbitrarily dense, if the frequency of interactions is sufficiently small. Turning the sparsity of the causality relationship into a \emph{soft} requirement is a major departure from the LLL setting and our main technical contribution. We do this by developing an entropy compression argument, in which we carefully amortize the entropy injected into the system to encode the effect of noise on the state trajectory. It is worth pointing out that even though our technique applies to the far more general noisy setting, in the absence of noise it recovers the main result of~\cite{JACM}, thus providing a smooth relationship between lack of internal conflict and robustness to noise.

\section{Formal Setting and Statement of Results}

In the absence of observational and environmental noise we can think of the state evolution under a policy $P$ as a random walk on a certain digraph on $\Omega$. Specifically, at each flawed state $\sigma \in \Omega$, for each action in the support of $P(O(\sigma))$, there is a bundle of outgoing arcs of total probability 1, corresponding to the state-transitions from $\sigma$ under this action. The convolution of $P(O(\sigma))$ with the distribution inside each bundle yields the state-transition probability distribution from each flawed state $\sigma$. 

The presence of observational and environmental noise both distorts the transition probabilities and introduces new transitions. For example, whenever observational noise causes $O(\sigma)$ to differ from the set of flaws truly present in $\sigma$, the controller may chose an action (from the support of $P(O(\sigma))$) under which there are transitions from $\sigma$ that were not present in the noise-free digraph. We model the overall distortion induced by noise by taking the noise-free digraph, which we think of as the \emph{principal} mechanism for state evolution, reducing the probabilities on all its edges uniformly by a factor of $1-p$, and allowing the leftover probability mass to be distributed arbitrarily, in order to form the noise. More precisely:
\begin{itemize}
\item
Let $D_{\principal}$ be the digraph on $\Omega$ of possible state-transitions under policy $P$, with a self-loop added at every flawless state. Let $\rho_{\principal}$ be the $P$-induced state-transition probability distribution, augmented so that all self-loops at flawless states have probability 1.
\item
Let $D_{\noise}$ be an {\bf arbitrary} digraph on $\Omega$. For each vertex $\sigma$ in $D_{\noise}$, let $\rho_{\noise}(\sigma,\cdot)$ be an {\bf arbitrary} probability distribution on the arcs leaving $\sigma$. 
\item
We will analyze the Markov chain on $\Omega$ which at every $\sigma \in \Omega$, with probability $p$ follows an arc in $D_{\noise}$ and with probability $1-p$ follows an arc in $D_{\principal}$. Formally, for every $\sigma \in \Omega$,
\begin{align*}
\rho(\sigma,\cdot)  =  (1-p) \cdot  \rho_{\principal}(\sigma, \cdot) + p \cdot  \rho_{\noise}( \sigma, \cdot ) \enspace.
\end{align*}
We assume that the system starts at a state $\sigma_1$, according to some unknown probability distribution $\theta$.
\end{itemize}

Requiring that the effect of noise is captured by a mixture of the original (principal) chain and an arbitrary chain is the only assumption that we make. In particular, by allowing $D_{\noise}$ and $\rho_{\noise}$ to be arbitrary we forego the need to posit specific models of observational and environmental noise, lending greater generality to our results. To see this, let $U(\sigma)$ denote the set of flaws actually present in $\sigma$ (and, slightly abusing notation, also the characteristic vector of $U(\sigma) \subseteq F$). In any step where the state transition distribution is not the principal one, we can think of this as occurring because $O(\sigma) \neq U(\sigma)$ and the distribution corresponds to $P(O(\sigma))$, or because $O(\sigma) \neq U(\sigma)$ and the distribution does not \emph{even} correspond to $P(O(\sigma))$, or because $O(\sigma) = U(\sigma)$ but, silently, the distribution followed is different from $P(O(\sigma))$. In particular, notice that whenever $O(\sigma) = \mathbf{0}$, the controller thinks  she has arrived at a flawless state and, thus, authorizes a self-loop with probability 1. In such a case, the fact that the system will follow $\rho_{\noise}$ with probability $p$ means that we are allowing the noise not only to trick the controller to inactivity but also to silently move the system to a new state. Similarly, after the system arrives at a flawless state, i.e., $U(\sigma) = \mathbf{0}$, with probability $p$ it will then follow an arc in $D_{\noise}$, potentially to a flawed state. We allow this to occur to be consistent with (i) the idea that observational noise can occur at any state, even a flawless one, thus causing unneeded, potentially detrimental action, and (ii) with the idea that flaws can be introduced spontaneously from the environment at any state. Our goal is, thus, to prove that from \emph{any} initial state, after a small number of steps, the system will reach a flawless state, despite the noise. As we will see, what will matter about the noise is the extent to which noise-induced transitions introduce flaws in the state.

Let $D = D_{\principal} \cup D_{\noise}$. To avoid certain trivialities we will assume that there exists a constant $B < \infty$ such that $2^{-B} < \rho(\sigma,\tau) < 1-2^{-B}$ for every arc $(\sigma,\tau) \in D$. For each state $\sigma$, we denote the highest priority flaw present in $\sigma$ by $\pi(\sigma)$;  if $\pi(\sigma) = f_i$, we label all arcs leaving $\sigma$ as $\sigma \xrightarrow{i} \cdot$, i.e., with the index of the flaw to which we attribute the transition (we use $i$ instead of $f_i$ as the label to lighten notation). We will refer to $\pi(\sigma)$ as the flaw \emph{addressed} at $\sigma$.
\begin{causality}
For an arc $\sigma \xrightarrow{i}  \tau$ in $D$ and a flaw $f_j$ present in $\tau$ we say that $f_i$ causes $f_j$ if  $f_j \not\ni \sigma$. 
The digraph on $[m]$ where $i \rightarrow j$ iff $D$ contains an arc such that $f_i$ causes $f_j$ is the causality digraph $C(D)$. 
\end{causality}
\begin{neigh}
The \emph{neighborhood} of a flaw $f_i$ in $C=C(D)$ is $\Gamma (f_i) = \{ f_i \} \cup \{f_j : i \to  j \text{ exists in $C$}\}$. 
\end{neigh}

For our condition we will need to bound from \emph{below} the entropy injected into the system in each step. To that end we define the potential of each flaw $f_i$ to be
\begin{equation}\label{eq:pot_def}
\pot(f_i) = \min_{\sigma: \pi(\sigma) = f_i} H[ \rho(\sigma, \cdot) ]  \enspace .
\end{equation}
We extend the definition to sets of flaws i.e.,  $\pot(S) = \sum_{f \in S} \pot(f)$, where $\pot(\emptyset) = 0$. 

In the absence of noise, $\Pot(f_i)$ expresses a lower bound on the diversity of ways to address flaw $f_i$, by bounding from below the ``average number of random bits consumed" whenever $f_i$ is addressed. Thus, it bounds from below the rate at which the controller explores the state space \emph{locally}. 
The presence of noise may decrease or may increase the potential. For example, if all arcs in $D_{\noise}$ are self-loops, then the noise is equivalent to the action-buttons ``sometimes not working" and its only (and very benign) effect is to slow down the exploration by a constant factor. At the other extreme, if $D_{\noise}$ is the complete digraph on $\Omega$ and $\rho_{\noise}$ is uniform, then (unless $p$ is extremely small) the situation is, clearly, hopeless. Correspondingly, even though the potential has been greatly increased, the causality relationship is complete. We note that, trivially, the potential of each flaw is bounded from below by the minimum entropy injected by the principal alone whenever the flaw is addressed, i.e., $\pot(f_i) \ge (1-p) \min_{\sigma : \pi(\sigma) = f_i} H[\rho_{\principal}(\sigma)]$.\smallskip

The other important characteristic of each flaw $f_i$ is its congestion, i.e., the maximum number of arcs with label $i$ that lead to the same state. For the same reason we would like the potential of a flaw to be big, we would like its congestion to be small: if arcs from different states in $f_i$ lead to the same state, then exploration slows down and the entropy injected into the system must be appropriately discounted in order to yield a good measure of the rate of state space exploration. To see this observe that $\pot(f_i)$ is independent of the destinations of the arcs leaving $f_i$ and compare the case where these destinations are all distinct with the case where they all lie in a small (bottleneck) set. As the congestion due to the principal and the congestion due to noise will have different effects, we need to account for them separately. Let $A_{\principal}(\sigma)$ denote the support of $\rho_{\principal}(\sigma, \cdot)$ and  $A_{\noise}(\sigma)$ denote the support of $\rho_{\noise}(\sigma, \cdot)$. 
\begin{backw}
For any flaw $f_i \in F$, let 
\begin{eqnarray*}
\Cong_{\principal}(f_i) 	& = & \max_{\tau \in \Omega} |\{\sigma \in f_i: \tau \in A_{\principal}(\sigma)\}| \\
\Cong_{\noise}(f_i) 		& = & \max_{\tau \in \Omega} |\{\sigma \in f_i: \tau \in A_{\noise}(\sigma)\}| 
\enspace .
\end{eqnarray*}
Let $\co_{\principal}^{f_i}  = \log_2 \Cong_{\principal}(f_i)$. Let $\co_{\noise}^{f_i}  = \log_2 \Cong_{\noise}(f_i)$. Let $\co_{\noise} = \max_{f_i \in F} \co_{\noise}^{f_i}$. 
\end{backw}

Let $C_{\principal}$ and $C_{\noise}$ be the causality graphs of $D_{\principal}$ and $D_{\noise}$, respectively, and let $\Gamma_{\principal}(f_i)$ and $\Gamma_{\noise}(f_i)$ be the corresponding neighborhoords.  Let $\Delta_i = \left|\Gamma_{\noise}  (f_i)\right|$. Recall that $h( p ) = - p \log_2 p  - (1-p) \log_2 (1 - p )$ is the \emph{binary entropy} of $p \in [0,1]$. To express the lost efficiency due to noise in addressing flaw $f_i$, we let
\begin{eqnarray*}
q_i(p) 	& = & p\left( \Delta_i\left( \co_{\noise}+\frac{5}{2}+h(p)\right) - 2 - h(p)\right) \\
		& \leq & p \, \Delta_i (\co_{\noise}+4) \enspace .
\end{eqnarray*}
Observe that $q_i(p)$ is independent of the policy and that its leading term is $p \Delta_i$. This means that, unlike the  LLL, the number, $\Delta_i$, of different flaws that may be introduced when addressing a flaw can be arbitrarily large if the \emph{frequency of interactions} between flaws, captured by $p$, is sufficiently small. Our main result establishes a condition under which the probability of not reaching a flawless state within $O(\log_2 |\Omega | +m)$ steps is exponentially small. To state it wefine for each flaw $f_i$,  
\[
\mathrm{Amenability}(f_i) =  \Pot(f_i) - \co_{\principal}^{f_i} \enspace .
\]

\begin{theorem}\label{Molloy}
 If for every flaw $f_i \in F$,
\begin{align}\label{eq:amenable}
\sum_{f_j \in \Gamma_{\principal}(f_i)}   2^{- 
\mathrm{Amenability}(f_j)+ q_j(p) }  < 2^{-(2+h(p))} \enspace ,
\end{align}
then there exists a constant $R>0$ depending on the slack in~\eqref{eq:amenable}, such that for every $s>1/2$, the probability of not reaching a flawless state after $Rs(\log_2 |\Omega | +m)$ steps is less than $\exp(-s)$.
\end{theorem}
\begin{remark}
In the noiseless case, i.e., when $p=0$, equation \eqref{eq:amenable} becomes an asymmetric LLL criterion. In particular, the main result of~\cite{JACM} is that if $\co_{\principal}^{f_i} = 0$ and all distributions $\rho_{\principal}(\sigma,\cdot)$ are uniform over their support $A_{\principal}(\sigma)$, then, a sufficient condition for reaching a flawless state quickly is that  for every $f_i \in F$, $\sum_{f_j \in \Gamma_{\principal}(f_i)} 1/a_j < 1/\mathrm{e}$, where $a_j = \min_{\sigma \in f_j: \pi(\sigma) = f_j} |A_{\principal}(\sigma)|$. We see that in this setting our condition~\eqref{eq:amenable} recovers this, up to the constant on the right hand side, i.e., $1/4$ vs. $1/\mathrm{e}$.
\end{remark}

\section{Related Work}\label{sec:related}

\subsection{POMDPs and the Reachability Problem} 

Markov Decision Processes (MDPs) are widely used models for describing problems in stochastic dynamical systems~\cite{CompetitiveMDP,Puterman,bertsekas2012}, where an agent repeatedly takes actions to achieve a specific goal while the environment reacts to these actions in a stochastic way. In an MDP the agent is assumed to be able to \emph{perfecty} observe the current state of the system and take action based on her observations. In a  \emph{partially} observable Markov Decision Process (POMDP) the agent only receives limited, and possibly inaccurate, information about the current state of the system. POMDPs have been used to model and analyze problems in artificial intelligence and machine learning such as reinforcement learning~\cite{chrisman1992reinforcement,RL}, planning under uncertainty~\cite{Planning}, etc. 

Formally, a  discrete POMDP is defined by the following primitives (all sets are assumed finite): (i) a \emph{state space} $\Omega$, (ii) a finite alphabet of \emph{actions}  $\cal A$, (iii) an observation space $\cal O$, (iv) an {\em action}-conditioned state transition model $\Pr(\tau | \sigma, a)$, where $\sigma,\tau \in \Omega$ and $a \in \cal A$, (v) an observation model $\Pr(o | \sigma)$, where $\sigma \in \Omega $ and $o \in \cal O$,  (vi) a cost function $\mathrm{c}: \Omega  \mapsto \mathbb{R}$ (or more generally a map from state-action pairs to the reals), and (vii) a desired criterion to minimize, e.g., expected total cumulative cost $\sum_{t=0}^\infty \ex \left[ \mathrm{c}(\sigma_t) \right]$, where $\sigma_t$ is the random variable that equals the $t$-th state of the trajectory of the agent. Finally, various choices of controllers are possible. For instance, a \emph{stochastic memoryless} controller is a map from the \emph{current} observation to a probability distribution over actions, whereas a \emph{belief-based} controller conditions its actions on probability distributions over the state space (i.e., beliefs) that are sequentially updated (using Bayes rule or some approximation of it) while the agent is interacting with the environment.

Unfortunately, the problem of computing an optimal policy for a POMDP, i.e., designing a controller that minimizes the expected cost, is highly intractable~\cite{PapaTsitsi,mundhenk2000complexity} and, in general, undecidable~\cite{madani1999undecidability}. Notably, the problem remains hard even if we severely restrict the class of controllers over which we optimize~\cite{PapaTsitsi, littman1998computational, etessami2010complexity, Nikos}. As far as we know, the only tractable case~\cite{Nikos} requires both the cost function and the class of controllers over which we optimize to be extremely restricted. In particular, the controller can not observe or remember anything and must apply the same distribution over actions in every step.

An important special case that has motivated our work is the \emph{reachability} problem for POMPDs. Here, one has a set of \emph{target} states $T \subseteq \Omega$, and the goal is to design a controller that starting from a distribution $\theta$ over $\Omega$, guides the agent to a state in $T$ (almost surely) with the optimal expected total cumulative cost. As shown in~\cite{Chater}, the problem is undecidable in the general case. In the same work, for the case where the costs are positive integers and the observation model is deterministic, i.e., the observations induce a partition of the state space, the authors give an algorithm which runs in time doubly-exponential in $|\Omega|$ and returns doubly-exponential lower and upper bounds for the optimal expected total cumulative cost, using a belief-based controller. On the other hand, our work establishes a sufficient condition for a stochastic memoryless controller to reach the target set $T$ rapidly (in time logarithmic in $|\Omega|$ and linear in $|F|$), in the case where each individual observation is binary valued (set membership) and the observation model is arbitrarily stochastic. To our knowledge, this is the first tractability result for a nontrivial class of POMDPs under stochastic memoryless controllers.

\subsection{Focusing and Prioritization}\label{sec:foc}

To achieve our results the controller must be focused and prioritize. The idea of focusing was introduced by Papadimitriou~\cite{papafocus} in the context of satisfiability algorithms, and amounts to ``if it ain't broken don't fix it", i.e., state evolution should only happen by changing the values of variables that participate in at least one violated constraint. One way to implement this idea is to always first select a violated constraint (flaw) and then take actions that tend to get rid of it. This has been an extremely successful idea in practice~\cite{walksat, circumspect} and it is often materialized by selecting a \emph{random} flaw to address in each step. We remark that our methods allow, in fact, also the analysis  of controllers that address a random flaw in each step, but for simplicity of exposition we chose to only present the case of a fixed permutation (prioritization).

Focusing is not only a good algorithmic idea, but also enables proofs of termination. Specifically, at the foundation of the argument of Moser and Tardos~\cite{MT} is the following observation: whenever an algorithm (focused or not) takes $t$ or more steps to reach a flawless state, say through flawed states $\sigma_1, \sigma_2, \ldots, \sigma_t$, there exists, by definition, a sequence of flaws $w_1, w_2, \ldots, w_t$ such that $\sigma_i \in w_i$. Therefore, by establishing a (potentially randomized) rule for selecting a flaw present in the state at each step, we can construct a random variable $W_t = w_1, w_2, \ldots, w_t$ to act as a \emph{witness} of the fact that the algorithm took at least $t$ steps. While, though, prima facie all constructions are equivalent, our capacity to bound the set of all possible such sequences is not. In particular, if the algorithm is focused and in each step we report the flaw on which the algorithm focused, then we can take advantage of the following observation: each appearance of a flaw $f_i$ in the witness sequence, with the potential exception of the very first, must be preceded by a distinct appearance of a flaw $f_j$ that causes $f_i$.  This allows us to bound the rate at which the entropy of the set of $t$-witness sequences grows with $t$. Of course, in a general setting, there is good reason to believe that prioritization, i.e., focusing on the flaw determined by a fixed permutation, will be not be the best one can do. In particular, observe that for the same $D_{\principal}$, different permutations $\pi$ give rise to different causality graphs. On the other hand, at the level of generality of this work, i.e., without any assumptions about the system at hand, we can not really hope for a more intelligent choice. 

\subsection{LLL algorithmization}

The Lov\'{a}sz Local Lemma (LLL)~\cite{LLL} is a non-constructive method for proving the \emph{existence} of flawless states that has served as a cornerstone of the probabilistic method. To use the LLL one provides a probability measure $\mu$ on $\Omega$, often the uniform measure, transforming flaws to (``bad") events, so that the existence of flawless states is equivalent to $\mu(\bigcup_{i=1}^m f_i) < 1$. The key quantity to control in order to prove this is negative dependence, i.e., the extent to which the probability of a bad event may be increased (boosted) by conditioning on the non-occurrence of other bad events. Roughly speaking, the LLL requires that for each bad event $f$, only a small number of other bad events should be able to boost $\mu(f)$ in this manner, whereas conditioning on the non-occurrence of all other bad events should not increase $\mu(f)$. Representing the boosting relationship in a graphical manner, with vertices corresponding to bad events pointing to their potential boosters, at a high level, the LLL requirement is that this digraph is sparse.

As one can imagine, whenever one proves that $\Omega$ contains flawless objects via the LLL it is natural to then ask if some such object can be found efficiently. Making the LLL constructive has been a long quest, starting with the work of Beck~\cite{beck_lll}, with subsequent works of Alon~\cite{alon_lll}, Molloy and Reed~\cite{mike_stoc}, Czumaj and Scheideler~\cite{Czumaj_lll}, Srinivasan~\cite{aravind_08} and others. Each of these works established a method for finding flawless objects efficiently, but with additional conditions relative to the LLL. A breakthrough was made by Moser~\cite{moser} who gave a very elegant algorithmization of the LLL for satisfiability via entropy compression. Very shortly afterwards, Moser and Tardos in a landmark paper~\cite{MT} made the LLL constructive for every product measure $\mu$. Specifically, they proved that if one starts by sampling an initial state according to $\mu$, and in every step selects an arbitrary occurring bad event and resamples its variables according to $\mu$, then with high probability a flawless state will be reached within $O(m)$ steps. 

Following~\cite{MT}, several works~\cite{szege_meet, SrinivasanPerm, AI, JACM, HV, Harmonic, Commu} have extended the scope of LLL algorithmization beyond product measures. In these works, unlike~\cite{MT}, one has to also provide either an explicit algorithm~\cite{szege_meet, SrinivasanPerm}, or an algorithmic framework~\cite{JACM,Harmonic,HV, Commu}, along with a way to capture the \emph{compatibility} between the algorithm's actions for addressing each flaw $f_i$ and the measure $\mu$.  
As was shown in~\cite{HV, Harmonic, Commu}, one can capture compatibility by  letting
\begin{align}\label{charge}
d_i   =  \max_{\tau \in \Omega } \frac{\nu_i(\tau)}{ \mu(\tau) } \ge 1 \enspace ,
\end{align} 
where $\nu_i(\tau)$ is the probability of ending up at state $\tau$ at the end of the following experiment: sample $\sigma \in f_i$ according to $\mu$ and address flaw $f_i$ at $\sigma$. 
An algorithm achieving $d_i = 1$ is a \emph{resampling oracle} for flaw $f_i$. If $d_i =1$ for every $i \in [m]$, then it was proven in~\cite{HV} that the causality digraph equals the boosting digraph mentioned above and the condition for success is identical to that of the LLL (observe that the resampling algorithm of Moser and Tardos~\cite{MT} is trivially a resampling oracle for every flaw). More generally, ascribing to each flaw $f_i$ the \emph{charge} $\gamma(f_i) =  d_i \cdot \mu(f_i)$, yields the following user-friendly algorithmization condition~\cite{Harmonic}, akin to the asymmetric Local Lemma: if for every flaw $f_i \in F$,
\begin{align} \label{LLL}
\sum_{ f_ j \in \Gamma(f_i)  } \gamma(f_j)  < \frac{1}{4} \enspace , 
\end{align}
then with high probability the algorithm will reach a sink after $O( \log | \Omega | + m )$ steps.  

Even though the noiseless case is only tangential to the main point of this work, as an indication of the sharpness of our analysis, we point out that in the noiseless case, our condition~\eqref{eq:amenable} is identical to~\eqref{LLL} with $\gamma(f_i)$ replaced by $\chi(f_i) :=   2^{-\Pot( f_i )  + b_{\principal}^{f_i}}$. In general, $\gamma(f_i)$ and $\chi(f_i)$ are incomparable. Roughly speaking, settings where $b_{\principal}^{f_i}$ is small and $d_i$ is large favor $\chi(f_i)$ over $\gamma(f_i)$ and vice versa, while the two meet when $b_{\principal}^{f_i}=0$, $\mu$ is uniform, and the transition probabilities are uniform,  as in~\cite{JACM}.

In terms of  techniques, as hinted in Section~\ref{sec:foc}, proofs of LLL algorithmizations consist of two independent parts. In one part, one bounds from above the probability of any witness sequence occurring, or in the case of Moser's entropic argument, bounds from below the entropy injected to the system while addressing the sequence. In the other part, one has to estimate the [entropy of the] set of possible witness sequences, using syntactic properties considerations mandated by causality: roughly speaking every occurrence of a flaw in a witness sequence, with the potential exception of the very first, must be preceded by an occurrence of some flaw that causes it. Finally, one compares the rate at which the probability of a $t$-step witness sequence decreases (or the rate at which entropy is increased) with the rate at which the [entropy of the] set of possible witness sequences increases, to establish that their product tends to 0 with $t$.
 
In this paper, exactly because we aim to capture the intensity of interactions between flaws under adversarial noise, we need to take a different approach. In particular, our proof can be thought of as entangling the two parts described above in order to establish that, while adversarial noise can make the imposed syntactic requirements inherited by the causality graph very weak (by making the graph extremely dense), the fact that the intensity of the noise is low, suffices to control the growth rate of the entropy of the set of witness sequences. The result is a carefully tuned argument that amortizes the entropy injected into the system against its effect on the entropy of the set of Break Forests. Key to the capacity to perform this amortization is the use of so-called Break Forests, introduced in~\cite{AI}, which localize in time the introduction of new flaws in the state. This property of Break Forests was not used in earlier works~\cite{AI,Harmonic} and allows us to use a different amortization for the flaws introduced by the principal vs.\ those introduced by noise.

\section{Termination via Compression}

Our analysis will not depend in any way on the distribution $\theta$ of the initial state. As a result, without loss of generality, we can assume that the process starts at an arbitrary but fixed state $\sigma_{\mathrm{init}}$. We let $A(\sigma)$ denote the support of $\rho(\sigma, \cdot)$, i.e., $A(\sigma)$ is the set of all states reachable by the process in a single step from $\sigma$.

\begin{definition}
We refer to the (random) sequence $\sigma_{\mathrm{init}} = \sigma_1,\ldots,\sigma_{t+1}$, entailing the first $t$ steps of the process, as the \emph{$t$-trajectory}. A $t$-trajectory is \emph{bad} iff $\sigma_1,\ldots,\sigma_{t+1}$ are all flawed.
\end{definition}

We model the set of all possible trajectories as an infinite tree whose root is labelled by $\sigma_1 = \sigma_{ \mathrm{init} }$.  The root has $|A(\sigma_{\mathrm{init}})|$ children corresponding to (and labelled by) each possible value of $\sigma_2$. More generally, a vertex labeled by $\sigma$ has $|A(\sigma)|$ children, each child labeled by a distinct element of $A(\sigma)$, i.e., a distinct possible value of $\sigma_{i+1}$. Every edge of this infinite vertex-labelled tree is oriented away from the root and labelled by the probability of the corresponding transition, i.e., $\rho(\sigma,\tau)$, where $\sigma$ is the parent and $\tau$ is the child vertex.  By our assumption, every such edge label is at least $2^{-B}$.

We call the above labelled infinite tree the \emph{process tree} and note that it is nothing but the unfolding of the Markov chain corresponding to the state-evolution of the process. In particular, for every vertex $v$ of the tree, the probability, $p_v$, that an infinite trajectory will go through $v$ equals the product of the edge-labels on the root-to-$v$ path. In visualizing the process tree it will be helpful to draw each vertex $v$ at Euclidean distance $-\log_2 p_v$ from the root. This way all trajectories whose last vertex is at the same distance from the root are equiprobable, even though they may entail wildly different numbers of steps (this also means that sibling vertices are not necessarily equidistant from the root). Finally, we color the vertices of the process tree as follows. For every infinite path that starts at the root determine its maximal prefix forming a bad trajectory. Color the vertices of the prefix red and the remaining vertices of the path blue.

In terms of the above picture, our goal will be to prove that there exist a critical radius $x_0$ and $\delta >0$, such that the proportion of red states at distance $x_0$ from the root is at most $1-\delta$. Crucially, $x_0$ will be polynomial, in fact linear, in $m=|F|$ and $ \log_2 | \Omega | $. Since we will prove this for every possible initial state and the process is Markovian, it follows that the probability that the process reaches distance $x$ from the root while going only through red states is at most $(1-\delta)^{\lfloor x/x_0 \rfloor}$. 

To prove that red vertices thin out as we move away from the root we stratify the process tree as follows. Fix any real number $x>0$ and on each infinite path from the root mark the first vertex of probability $2^{-x}$ or less, i.e., the first vertex that has distance at least $x$ from the root. Truncate the process tree so  that the marked vertices become  leaves of a finite tree. Let $L(x)$ be the set of all root-to-leaf paths (trajectories)  in this finite tree  and let $B(x) \subseteq L(x)$ consist of the bad trajectories. Now, let $I$ be the random variable equal to an infinite trajectory of the process and let $\Sigma = \Sigma(x)$ be the random variable equal to the prefix of $I$ that lies in $L(x)$. By definition, $\sum_{\ell \in L(x)} \Pr[\Sigma = \ell] = 1$, while $\Pr[\ell] \in (2^{-x-B},2^{-x}]$ for every $\ell \in L(x)$, since $-\log_2 \rho \ge B$. Let $P = P(\Sigma)$ be the maximal red prefix of $\Sigma$ and observe that if $\Sigma \in B(x)$ then $P = \Sigma$. Therefore,
\begin{equation}\label{beauty}
H[P] \; \ge \sum_{\ell \in B(x)} \Pr[\Sigma = \ell] (-\log_2 \Pr [ \Sigma =\ell]) \; \ge \; x  \sum_{\ell \in B(x)} \Pr[\Sigma = \ell] \; = \; x \Pr[\Sigma \in B(x)] \enspace .
\end{equation}
Assume now that there exist $M_0 > 0$ and $\lambda < 1$, such that $H[P] \le \lambda  x + M_0$, for every $x>0$. Then~\eqref{beauty} implies that for $x_0 = 2 M_0/ (1 - \lambda )$, 
\begin{equation}\label{eq:posi}
\Pr[\Sigma \in B(x_0)] \le \frac{H[P]}{x_0} \le \frac{ \lambda x_0 + M_0}{x_0} = \lambda+ \frac{1- \lambda}{2} = \frac{1 + \lambda}{2} < 1 \enspace .
\end{equation}

If $\Sigma \in B(x_0)$,  we treat the reached state as the root of a new finite tree and repeat the same analysis, as it is independent of the starting state. It follows in this manner that for every integer $T \ge 1$, the probability that the process reaches a state at distance $T(x_0 + B)$ or more from the root by going only through red states is at most $\left((1+\lambda)/2\right)^T$. Thus, for any $s > 1/2$, the probability that the process reaches a state at distance
\[
E = \left\lceil\frac{2s}{1+\lambda} \right\rceil (x_0 + B) = O\left(\frac{s M_0}{1-\lambda^2} \right)
\]
or more from the root by going only through red states is at most $\left((1+\lambda)/2\right)^{\left\lceil \frac{2s} {1 +\lambda}\right\rceil} < \exp(-s)$. 

Since $\rho(\sigma,\tau) < 1 - 2^{- B}$, it follows that $ - \log_2 \rho(\sigma,\tau) 
>  2^{-B}$, for every arc in $D$. Thus, after $ 2^{B} E $ steps the process is always at distance $E$ or more from the root. Thus, the probability of not reaching a flawless state after $ 2^{B} E=  O\left(\frac{s M_0}{1-\lambda^2} \right)$ steps is $\exp(-s)$. Therefore Theorem~\ref{Molloy} follows from the following.

%
%
%
 
\begin{theorem}\label{Amenability}
Let $\Xi = \max \{  \co_{\noise },  \co_{\principal}  \}$ and $\Delta = \max_{j \in F} \Delta_j$. If there exists $\lambda < 1$ such that for all $j \in [m]$,
\[
\sum_{f_i \in \Gamma_{\principal}(f_j)}   2^{- (\lambda \Pot(f_i)  -   \co_{\principal}^{f_i} - q_i(p) )}  < 2^{-(2+h(p))} \enspace ,
\]
then $H[P] \le \lambda x + M_0$ for every $x>0$, where $M_0   = \log_2 |\Omega |  + m  (\Delta+1)(\Xi+4) +  \lambda B$.
\end{theorem}

\section{Break Sequences}

Recall that $\pi$ is an arbitrary but fixed ordering of the set of flaws $F$ and that the highest flaw present in each state $\sigma$ is denoted by $\pi(\sigma)$. We will refer to $\pi(\sigma)$ as the flaw \emph{addressed} at state $\sigma$, i.e., as in the noiseless case, even though  the action distribution $P(O(\sigma))$  may be ``misguided" whenever $O(\sigma) \neq U(\sigma)$.

\begin{definition}
Given a bad $t$-trajectory $\Sigma$, its \emph{witness} sequence is $W(\Sigma) = w_1, \ldots, w_t = \{\pi(\sigma_i)\}_{i=1}^t$.
\end{definition}

To prove Theorem~\ref{Amenability}, i.e., to gain control of bad trajectories and thus of $H[P]$, we introduce the notion of \emph{break sequences} (see also~\cite{AI,Harmonic}). Recall that $U(\sigma)$ denotes the set of flaws present in $\sigma$.
\begin{definition}
Let  $B_0 = U(\sigma_1)$. For $1 \le i\le t-1$, let $B_i = U(\sigma_{i+1}) \setminus ( U (\sigma_i) \setminus w_i )$.
\end{definition}
Thus, $B_i$ is the set of flaws ``introduced" during the $i$-th step, where if a flaw is addressed in a step but remains present in the resulting state we say that it ``introduced itself". Each flaw $f \in B_i$ may or may not be addressed during the rest of the trajectory. For example, $f$ may get fixed ``collaterally" during some step taken to address some other flaw, before the controller had a chance to address it. Alternatively, it may be that $f$ remains present throughout the rest of the trajectory, but in each step $i < j \le t-1$ some other flaw has greater priority than $f$. It will be crucial to identify and focus on the subset of flaws $B_i^* \subseteq B_i$ that \emph{do} get addressed during the $t$-trajectory, causing entropy to enter the system. Per the formal Definition~\ref{def:bs} below, the set of such flaws is $B_i^*  = B_i \setminus \{O_i \cup N_i \}$, where $O_i$ comprises any flaws in $B_i$ that get eradicated collaterally, while $N_i$ comprises any flaws in $B_i$ that remain present in every subsequent state after their introduction without being addressed. 
\begin{definition}\label{def:bs}
The \emph{Break Sequence} of a $t$-trajectory is $B_0^*, B_1^*, \ldots, B_{t}^*$, where for $0 \le i \le t$,
\begin{align*}
B_i^*  	& = B_i \setminus \{O_i \cup N_i \} \enspace , \text{where}\\
O_i 		& =  	\{f \in B_i \mid  \exists  j \in [i+1,t] : 
f \notin U(\sigma_{j+1})  \wedge  \forall \ell \in [i+1,j]:   f \ne w_{\ell} \} \enspace , \\
N_i 		& =  	\{f \in B_i \mid \forall j \in [i+1, t] :  
f \in 	 U(\sigma_{j+1})  \wedge  \forall \ell \in [i+1,t]:   f \ne w_{\ell} \} \enspace .\
\end{align*}
\end{definition}

Given $B_0^*,B_1^*,\ldots,B_{i-1}^*$ we can determine the sequence $w_1, w_2, \ldots, w_i$ of flaws addressed inductively, as follows. Define $E_1 = B_0^*$, while for $i \ge 1$, let
\begin{equation}\label{eq:ri}
E_{i+1} = (E_{i} - w_i) \cup B_i^*  \enspace .
\end{equation}
Observe that, by construction, $E_i \subseteq U(\sigma_i)$ and $w_i \in E_i$. Therefore, for every $i$, the highest flaw in $E_i$ is $w_i$.

\section{Proof of Theorem~\ref{Amenability}}

For the analysis, we will assume that the state-transition distribution $\rho$ is realized in each step by flipping a coin with bias $p$ to determine if the state transition will occur according to $\rho_{\principal}(\sigma, \cdot)$ or $\rho_{\noise}(\sigma, \cdot)$. Let $I$ be the random variable equal to an infinite trajectory of the algorithm. For any fixed real number $x>0$, we define the following random variables. 
\begin{itemize}
\item
Let $\Sigma$ be the prefix of $I$ in $L(x)$.
\item
Let $P = \sigma_1, \sigma_2,  \ldots, \sigma_{Z+1}$ be the maximal bad prefix of $\Sigma$. Thus, $P$ consists of $Z$ steps.
\item
For $ i  \ge  1$:
\begin{itemize}
\item 
Let $\sigma'_i = \sigma_i$ for $i \le Z$, while $\sigma'_i = \emptyset$ for $i > Z$.
\item
Let $r_i = w_i(\Sigma)$ for $i \le Z$, while $r_i = \emptyset$ for $i > Z$.
\item 
Let $n_i$ be indicator r.v.\ that $\rho_{\noise}$ was employed in the $i$-th step of $P$, while $n_i = 0$ for $i  > Z$. 
\end{itemize}

\item Let $N 
= n_1, n_2, \ldots$ 

\item
Let $Y = Y(P) = B_0^*, B_1^*,\ldots,$ be the break sequence of $P$, where $B^*_i = \emptyset$ for $i  > Z$.
\item
Let $Y_1$ be the suffix of $Y$ starting at $B_1^*$.
\item
Let $L = |B_0^*|, |B_1^*|, \ldots$ 
\end{itemize}

Observe that $L$ determines $Z$ since $\sum_{i=0}^t |B_i^*| \ge t$ for all $t \le Z$, with equality holding only for $t = Z$. To pass from~\eqref{full_seq} to~\eqref{pointwise} we use that $L$ determines $Z$ and that $Y$ determines the sequence $r_1,r_2,\ldots, r_{Z}$. To pass from~\eqref{verylong} to~\eqref{poules} we use that there is a 1-to-1 correspondence between the elements of the witness sequence $r_1, \ldots, r_{Z}$ and the $Z$ elements in the sets $B_0^*, B_1^*, \ldots$ Thus,
\begin{align}
H[P] & =   H[Z, \sigma_1, \sigma_2,  \ldots, \sigma_{Z+1} ]  		\nonumber  				\\
	& \le H[L, \sigma_1,  \sigma_{2}, \ldots, \sigma_{Z+1}, Y , N] \nonumber   				\\
	& \le H[B_0^*, L] + H[N \mid L] +  H[Y_1 \mid B_0^*, L, N] + H[\sigma_{Z+1}] 
			+ \sum_{i  \ge 2 } H[\sigma'_{i-1} \mid \sigma'_i, Y, N, L] 		\label{full_seq} \\
	& \le H[B_0^*, L] + H[N \mid Z] +  H[Y_1 \mid B_0^*, L, N] + \log_2 |\Omega|
			+ \sum_{i  \ge 2 } H[\sigma'_{i-1} \mid \sigma'_i, r_{i-1},n_{i-1} , Z] \label{pointwise} \\
	& \le H[B_0^*, L] + h(p) \cdot \ex Z + H[Y_1 \mid B_0^*, L, N] + \log_2 |\Omega| 
			+ \sum_{i \ge 2 } \ex \left[ (1-p)  \co_{\principal}^{r_{i-1}} + p  \co_{\noise}^{r_{i-1}} \right]\label{verylong} \\
	&  \le   H[B_0^*, L] + h(p) \cdot \ex Z + H[Y_1 \mid B_0^*, L, N] + \log_2 |\Omega| 
			+ \sum_{i  \ge 0} \ex \mathrm{In}(B_i^*)	\label{poules}	 \enspace ,
\end{align}
where, recalling that $\co_{\noise} = \max_{j \in [m]} \co_{\noise}^{f_j}$, we define for an arbitrary set of flaws $S$, 

\begin{eqnarray}
\mathrm{In}(S) & = &  (1-p)  \sum_{f \in S }  \co_{\principal}^{f}  + p  |S|  \co_{\noise} \nonumber \\
& := & (1-p)  \mathrm{In}_{\principal}(S)	+ p \mathrm{In}_{\noise}(S)	 \label{eq:indef} \enspace .
\end{eqnarray}

To bound the right hand side of~\eqref{poules} we prove at the end of this section Lemmata~\ref{struct_code} and~\ref{meatless} presented below. In the rest of this section, all sums over index $i$ are sums over $i \ge 1$. 
\begin{lemma}\label{struct_code}
$H[B_0^*,L] \le m + 2\ex Z - \ex |B_0^*|$, and $\sum_i \ex |B_i^*| =  \ex Z   - \ex |B_0^*|$.
\end{lemma}
Using the chain rule for entropy to write $H[Y_1 \mid B_0^*, L, N]  = \sum_{i} H[B_i^* \mid B_0^*, \ldots, B_{i-1}^*, L, N]$ and combining the inequality in Lemma~\ref{struct_code} with~\eqref{poules}, we see that $H(P)$ is bounded from above by
\[
m + 2 \ex Z - \ex |B_0^*| + h(p) \cdot \ex Z + \sum_{i} H[B_i^* \mid B_0^*, \ldots, B_{i-1}^*, L, N] + \log_2 |\Omega| + \ex \mathrm{In}(B_0^*)  + \sum_i \ex \mathrm{In}(B_i^*)	
\]
Using the equality in Lemma~\ref{struct_code} to express $\ex Z$ as a sum, we see that the line above is equal to
\begin{align*}
m + & (1+h(p))\ex |B_0^*|  + \log_2 |\Omega| + \ex \mathrm{In}(B_0^*)  \\
 + &\sum_{i} 
	\left\{
		H[B_i^* \mid B_0^*,\ldots, B_{i-1}^*, L, N] + 
		\ex
			\left[
				\mathrm{In}(B_i^*) + (2+h(p)) |B_i^*|
			\right]
	\right\}	
	\enspace 
	.
\end{align*}

For an arbitrary set of flaws $S$, we define 
\begin{eqnarray*}
q(S) & = & \sum_{ f_j \in S}  q_j(p) \enspace , \\
g(S) & = & \lambda^{-1} (p( 2 + h(p)) |S| + q(S)) \enspace .
\end{eqnarray*}

\begin{definition}
Say that $C_{\principal}$ is $\lambda$-amenable if the conditions of Theorem~\ref{Amenability} are satisfied.
\end{definition}

\begin{lemma}\label{meatless}
If $C_{\principal}$ is $\lambda$-amenable, then 
\begin{align*}
\sum_i \{ H[B_i^* \mid B_0^*,\ldots,B_{i-1}^*, L, N] +\ex \left[\mathrm{In}(B_i^*) + (2+h(p)) |B_i^*|   \right] \} \le \lambda  (x+B + \ex g(B_0^*)) \enspace.
 \end{align*}
\end{lemma}
Lemma~\ref{meatless} thus implies that under the conditions of Theorem~\ref{Amenability},
\[
H(P) \le m +(1 + h(p) )  \ex |B_0^*| + \log_2 |\Omega| + \ex \mathrm{In}(B_0^*) + \lambda(x+B + \ex g(B_0^*)) \enspace .
\]
Recall that  $\Xi = \max \{  \co_{\noise },  \co_{\principal}\}$ and that $\Delta = \max_{j \in F} \Delta_j$. Since $\ex |B_0^*| \le m$ and $\ex \mathrm{In}(B_0^*) \le m \Xi$, we conclude, as claimed in Theorem~\ref{Amenability}, that $H(P) \le M_0 +\lambda x$, where
\begin{eqnarray*}
 M_0  
 & = &  \log_2 |\Omega | + m  \left( 2 +h(p)     +  \Xi \right)   +  \lambda B + \lambda \ex g(B_0^*) \\
 & \le & \log_2 |\Omega | + m  \left( 2 +h(p)    +  \Xi  \right)  +  \lambda B + [p (2+h(p))+ \max_{j \in F} q_j(p)] m \\
 & \le & \log_2 |\Omega | + m  \left( 6  + \Xi   +  \max_{j \in F} q_j(p)\right)  +  \lambda B  \\
 & \le & \log_2 |\Omega | + m  (\Delta+1)(\Xi+4) +  \lambda B  \enspace . \\
 \end{eqnarray*}

\begin{proof}[Proof of Lemma~\ref{struct_code}]
We will represent $B_0^*,L$ as a binary string $s$ of length $m+2Z-|B_0^*|$. Since $B_0^* \subseteq F$ the first $m$ bits of $s$ are the characteristic vector of $B_0^*$. We encode $L$ immediately afterwards, representing the $i$-th element of $L$, for each $i \in [Z]$, as $1^{|B_i^*|}0$. Decoding, other than termination, is trivial: after reading the first $m$ bits of $s$, the rest of the string is interpreted in blocks of the form $1^*0$. To determine termination we note that, by construction, $|B_0^*| + \sum_{i=1}^j |B_i^*| - j \ge 0$ for every $j \in [Z]$ with equality holding only for $j=Z$. Therefore, decoding stops as soon as equality holds for the first time. The representation of $L$ in this manner consists of $\sum_i |B_i^*|$ ones and $Z$ zeroes, i.e., of $2Z - B_0^*$ bits, since $\sum_i |B_i^*| = Z - |B_0^*|$.

For $\sum_{i} \ex |B_i^*|$ the claim follows readily from the fact $\sum_i |B_i^*| = Z - |B_0^*|$.
\end{proof}

For an arbitrary set of flaws $S$, we define
\[
\Pot^{-}(S) \; = \; \Pot(S) - g(S)  \enspace .
\]
Lemma~\ref{meatless} follows trivially by combining Lemmata~\ref{meat} and \ref{neato} below.
\begin{lemma}\label{meat}
If $C_{\principal}$ is $\lambda$-amenable, then for every $i \ge 1$,
\begin{align*}
 H \left[B_i^* \mid B_0^*,\ldots,B_{i-1}^*, L, N \right] +\ex[\mathrm{In}(B_i^*) + (2+h(p)) |B_i^*|]  
 \le \lambda  \ex[\Parent(B_i^*) +  g(r_i)] \enspace.
 \end{align*}
\end{lemma}
The proof of Lemma~\ref{meat} is presented in Appendix~\ref{sec:meat}. 
\begin{lemma}\label{neato}
$\sum_i \ex [  \Parent(B_i^*)  + g(r_i)  ]\le   \left(x+B \right) +\ex g(B_0^*) $.
\end{lemma}
\begin{proof}[Proof of Lemma~\ref{neato}] 
Since there is a 1-to-1 correspondence between the elements of the witness sequence $r_1, \ldots, r_{Z}$ and the $Z$ elements in the sets $B_0^*, B_1^*, \ldots$
\begin{eqnarray}\label{xesemas}
\sum_{i}  \pot(r_i) & = & \pot(B_0^*) +  \sum_{i} \pot(B_i^*)  \nonumber \\
 & \ge & \sum_{i} \left[ \pot(B_i^*)  - g(B_i^*) + g(r_i)  \right] -   g(B_0^*) \nonumber  \\
& = & \sum_i  [ \Parent(B_i^*) + g(r_i)  ]  -  g(B_0^*) \label{terrible} \enspace.
\end{eqnarray}

The chain rule for entropy gives~\eqref{eq:chain_chain_chain}. Since the evolution of $\Sigma$ is Markovian, inequality~\eqref{eq:almost} would have been an equality if it were not for the possibility that $r_i = \emptyset$. Finally, inequality~\eqref{eq:hahaha} follows from the definition of potential~\eqref{eq:pot_def}. Thus, 

\begin{eqnarray}
H[\Sigma] & = & \sum_{i}H[ \sigma_{i+1} \mid \sigma_{i} ] \label{eq:chain_chain_chain} \\
& = & \sum_{i} \sum_{\sigma \in \Omega } \Pr[ \sigma_{i} = \sigma ] \cdot  H[ \sigma_{i+1} \mid \sigma_i = \sigma ]  \nonumber \\
& \ge & \sum_{i}  \sum_{j \in [m]} \Pr[r_i   = f_j ] \sum_{\sigma \in f_j}  \Pr[ \sigma_i = \sigma \mid r_i = f_j ] \cdot H[ \rho(  \sigma, \cdot ) ] \label{eq:almost} \\
& \ge &  \sum_{i}  \sum_{j \in [m]} \Pr[ r_i = f_j] \cdot    \pot(f_j) \label{eq:hahaha} \\
& = & \sum_{i} \ex \pot(r_i) \enspace .\label{trajectory}
\end{eqnarray}
Combining~\eqref{xesemas} and~\eqref{trajectory} with the fact $H[ \Sigma ] \le x + B$ yields the lemma.
\end{proof}

\section{Proof of Lemma~\ref{meat}}\label{sec:meat}

We need to prove that if $C_{\principal}$ is $\lambda$-amenable, then for every $i \ge 1$,
\begin{equation}
 H \left[B_i^* \mid B_0^*,\ldots,B_{i-1}^*, L, N \right] +\ex[\mathrm{In}(B_i^*) + (2+h(p)) |B_i^*|]  
 \le \lambda  \ex[\Parent(B_i^*) +  g(r_i)] \enspace. \label{july4}
 \end{equation}
 
Recall that  $B_i^* = r_i = \emptyset$ for $i > Z$. Since $L$ determines $Z$, it follows that the conditional entropy in the left hand side of~\eqref{july4} is 0 for $i>Z$. Thus, \eqref{july4} holds trivially for $i>Z$ since $\mathrm{In}(\emptyset) = \pot^-(\emptyset) = g(\emptyset) = 0$.

In the rest of this section we consider an arbitrary but fixed $1 \le i \le Z$. For any such $i$, recall that $B_0^*, B_1^{*}, \ldots,B_{i-1}^{*}$ determine $r_i$ and $L$ determines $|B_i^*|$. Therefore,
\begin{eqnarray*} 
H[B_i^* \mid B_0^*,\ldots,B_{i-1}^*, L, N]  &\le&  H[B_i^* \mid r_i , |B_i^*|, n_i ] \nonumber \\
									 & =&   (1-p) \, H[ B_i^* \mid r_i, |B_i^*| , n_i = 0]  + p  \, H[ B_i^* \mid r_i, |B_i^*|, n_i = 1]       \enspace .
\end{eqnarray*}

For $j \in [m]$ let us denote $ \Pr_j [\cdot] =  \Pr \left[  \cdot \mid r_i = f_j \right]$, $ \ex_j [ \cdot  ] = \ex \left[  \cdot \mid r_i = f_ j  \right]$ and $ H_j[ \cdot  ] = H \left[  \cdot \mid r_i = f_ j  \right]$. With this notation, and recalling~\eqref{eq:indef}, we see that~\eqref{july4} follows from the inequalities in the following lemma, i.e., by multiplying each inequality with the probability that $r_i = f_j$ and summing up over $j \in [m]$. 
\begin{lemma}\label{ns_lemma}
For every $j \in [m]$,
\begin{align*}
H_j [B_i^* \mid n_i = 0 , |B_i^*|] 
+ \ex_j[\mathrm{In}_{\principal}(B_i^*) ]  
+ (2 + h(p)) \, \ex_j \left[ |B_i^*| \mid n_i = 0 \right] 
& \le  (1-p)^{-1} \lambda \ex_j \pot^{-}(B_i^*) \\
\\
H_j[ B_i^* \mid  n_i = 1, |B_i^*|]  
+ \ex_j[\mathrm{In}_{\noise}(B_i^*) ]  
+ (2 + h(p)) \, \ex_j \left[ |B_i^*| \mid n_i = 1 \right]  
& \le p^{-1} \lambda  g(f_j) \enspace .
\end{align*}
\end{lemma}

\begin{proof}
To lighten notation, let $p_k^c  = \Pr_j[ n_i  = c, |B_i^*| =k ]$ and $p_k^c(S) = \Pr_j[B_i^* = S \mid n_i = c , |B_i^*| = k]$. 

We start with the simpler case $c=1$. For any $j \in [m]$, recalling that $\log_2 \binom{\Delta_j}{k} \le \Delta_j \cdot h(k/\Delta_j)$, we get
\begin{eqnarray}
H_j[ B_i^* \mid  n_i = 1, |B_i^*|]  
+ \ex_j[\mathrm{In}_{\noise}(B_i^*) ]  
+ (2 + h(p)) \, \ex_j \left[ |B_i^*| \mid n_i = 1 \right]  
& \le & \nonumber \\
\sum_{k }  p_k^1  \bigg(   H_j[B_i^* \mid n_i = 1 , |B_i^*| = k]  +  k (\co_{\noise} + 2+h(p))   \bigg)   &\le & \nonumber \\
 \sum_{k} p_k^1  \cdot  \left(  \Delta_j\,h (k/\Delta_j) + k (  \co_{\noise} + 2+ h(p) ) \right)  \nonumber   & \le & \nonumber   \\ 
 \max_{k \in [0,\Delta_j]}  \left\{  \Delta_j\,h (k/\Delta_j) + k (  \co_{\noise} + 2+ h(p) ) \right\}  & < & \nonumber \\
\Delta_j\left( \co_{\noise}+\frac{5}{2}+h(p)\right) \nonumber & = & \\
 p^{-1} \left( q_j(p) + p (2+h(p))\right)  \label{july7} & , & 
 \end{eqnarray}
since $q_j(p) = p\left( \Delta_j\left( \co_{\noise}+\frac{5}{2}+h(p)\right) - 2 - h(p)\right)$.

For the case $c=0$ we need some preparation.  Observe that $C_{\principal}$ being $\lambda$-amenable implies that $\lambda \Pot(f_i)  - q_i(p) \ge 2 + h(p)$ for every flaw $f_i$, as otherwise~\eqref{eq:amenable} would be violated. Therefore, we see that for every set $S  \subseteq F$,
\begin{eqnarray}
\lambda \Pot(S) & \ge & q(S) + |S| \left(  2+ h(p)\right) \label{eq:july6} \\
\Pot^-(S) & \ge & 0 \enspace , \label{eq:potika}
\end{eqnarray}
where~\eqref{eq:potika} follows from~\eqref{eq:july6} since $\lambda \Pot^{-}(S) = \lambda \Pot(S) - q(S) - p|S| (2+h(p))$. The positivity of $\Pot^-$ also implies that for every set $B_i^*$,
\begin{eqnarray} 
\pot^{-} (B_i^*) & = &  \pot^{-} 
\left(
	B_i^* \cap \Gamma_{\principal}(r_i)
\right) + 
\pot^{-} \left(B_i^* \cap \overline{\Gamma_{\principal}(r_i)}\right) \\
& \ge & 
\pot^{-} \left(B_i^* \cap \Gamma_{\principal}(r_i)\right) \enspace .\label{eq:panagos}
\end{eqnarray}
Finally, for $j \in [m]$, we let $\mathcal{S}_k^{j}$ denote the set of all $k$-subsets of $\Gamma_{\principal} ( f_j )$. 

We can now start working towards our goal, which is to prove that for every $j \in [m]$ the first line below is non-negative. We start by invoking~\eqref{eq:panagos} to prove the first inequality below and~\eqref{eq:july6} to prove the second.
\begin{align*}
\lambda \ex_j \pot^{-}(B_i^*) + (1-p)\bigg(- H_j [B_i^* \mid n_i = 0 , |B_i^*|] - \ex_j[\mathrm{In}_{\principal}(B_i^*)] - (2 + h(p)) \, \ex_j \left[ |B_i^*| \mid n_i = 0 \right]\bigg) & \ge  \\ 
\lambda \ex_j \pot^{-}(B_i^* \cap \Gamma_{\principal}(r_i)) \\
+ (1-p) \bigg(- H_j [B_i^* \mid n_i = 0 , |B_i^*|] - \ex_j[\mathrm{In}_{\principal}(B_i^*)] - (2 + h(p)) \, \ex_j \left[ |B_i^*| \mid n_i = 0 \right]\bigg) & = \\
\sum_{k} p_k^0 \sum_{S \in \mathcal{S}_k^{j}} p_k^0(S)
	\left\{
		\lambda \pot^{-}(S) + (1-p) \bigg(\log_2 p_k^0(S) - \mathrm{In}_{\principal}(S) - (2 + h(p)) k \bigg)
	\right\}  
& = \\ 
\sum_{k} p_k^0 \sum_{S \in \mathcal{S}_k^{j}} p_k^0(S) 
	\left\{
		p \lambda \pot^{-}(S) 
		+ 
		(1-p) \bigg( 
				\lambda \pot^{-}(S)  + \log_2 p_k^0(S) -   \mathrm{In}_{\principal}(S) - (2 + h(p) )k 
			\bigg)
	\right\} 
& = \\ 
(1-p) \sum_{k} p_k^0 \sum_{S \in \mathcal{S}_k^{j}} p_k^0(S) 
	\left\{	  
		\frac{p\lambda \pot^{-}(S)}{1-p} + \lambda \pot^{-}(S) + \log_2 p_k^0(S) - \mathrm{In}_{\principal}(S) - (2+h(p)) k 
	\right\}  
& \ge \\
(1-p) \sum_{k} p_k^0 \sum_{S \in \mathcal{S}_k^{j}} p_k^0(S)  \left\{  \lambda \pot(S) -q(S)  +\log_2 p_k^0(S)  - \mathrm{In} _{\principal}(S) - (2 + h(p) )k  \right\}  &  \enspace .
\end{align*}
Letting $\zeta(S) = 2^{-  \lambda \pot(S) + q(S) + \mathrm{In}_{\principal}(S)  }$ it thus suffices to prove that  the left hand side of~\eqref{eq:lhso} is non-negative 
\begin{align}
\sum_{k} p_k^0 \left\{-(2 + h(p) )k +\sum_{S \in \mathcal{S}_k^{j}}  p_k^0(S) \log_2 \frac{p_k^0(S)}{\zeta(S)}\right\} \ge   & \label{eq:lhso} \\
\sum_{k} p_k^0 \left\{- (2 + h(p) )k- \log_2 \left( \sum_{S \in \mathcal{S}_k^{j}} \zeta(S) \right)\right\} \label{entropia} \enspace ,
\end{align}
where the inequality follows from the log-sum inequality. To prove that the right hand side of~\eqref{entropia} is non-negative we note that for all $k \geq 1$,
\begin{eqnarray*}
\sum_{S \in \mathcal{S}_k^{j}} \zeta(S) 
 \le
\left(\sum_{A \in \mathcal{S}_{k-1}} \zeta(A)\right) \left(\sum_{B \in \mathcal{S}_{1}} \zeta(B) \right)
\le \ldots  \le \left( \sum_{B \in \mathcal{S}_{1}} \zeta(B) \right)^k \enspace .
\end{eqnarray*}
Therefore, since $C_{\principal}$ is $\lambda$-amenable,
\begin{align*}
\sum_{S \in \mathcal{S}_k^{j}} \zeta(S)   \le  \left(\sum_{f_i \in \Gamma_{\principal}(f_j)}    2^{  -     \lambda \Pot(f_i) +  {\co_{ \principal }^{f_i}}   + q_i(p)   } \right)^k\le \left(2^{-(2+ h(p))}\right)^k \enspace .
\end{align*}
\end{proof}

\bibliographystyle{plain}

\bibliography{POMDP}

\begin{thebibliography}{10}

\bibitem{AI}
Dimitris Achlioptas and Fotis Iliopoulos.
\newblock Random walks that find perfect objects and the {L}ov{\'{a}}sz local
  lemma.
\newblock In {\em 55th {IEEE} Annual Symposium on Foundations of Computer
  Science, {FOCS} 2014, Philadelphia, PA, USA, October 18-21, 2014}, pages
  494--503. {IEEE} Computer Society, 2014.

\bibitem{Harmonic}
Dimitris Achlioptas and Fotis Iliopoulos.
\newblock Focused stochastic local search and the {L}ov{\'{a}}sz local lemma.
\newblock In {\em Proceedings of the Twenty-Seventh Annual {ACM-SIAM} Symposium
  on Discrete Algorithms, {SODA} 2016, Arlington, VA, USA, January 10-12,
  2016}, pages 2024--2038. {SIAM}, 2016.

\bibitem{JACM}
Dimitris Achlioptas and Fotis Iliopoulos.
\newblock Random walks that find perfect objects and the {L}ov{\'{a}}sz local
  lemma.
\newblock {\em J. {ACM}}, 63(3):22, 2016.

\bibitem{circumspect}
Mikko Alava, John Ardelius, Erik Aurell, Petteri Kaski, Supriya Krishnamurthy,
  Pekka Orponen, and Sakari Seitz.
\newblock Circumspect descent prevails in solving random constraint
  satisfaction problems.
\newblock {\em Proceedings of the National Academy of Sciences},
  105(40):15253--15257, 2008.

\bibitem{alon_lll}
Noga Alon.
\newblock A parallel algorithmic version of the local lemma.
\newblock {\em Random Struct. Algorithms}, 2(4):367--378, 1991.

\bibitem{beck_lll}
J{\'o}zsef Beck.
\newblock An algorithmic approach to the {L}ov\'asz local lemma. {I}.
\newblock {\em Random Structures Algorithms}, 2(4):343--365, 1991.

\bibitem{bertsekas2012}
Dimitri~P Bertsekas.
\newblock {\em Dynamic Programming and Optimal Control, Vol. II}.
\newblock Athena Scientific, 2012.

\bibitem{Chater}
Krishnendu Chatterjee, Martin Chmelik, Raghav Gupta, and Ayush Kanodia.
\newblock Optimal cost almost-sure reachability in {POMDP}s.
\newblock {\em Artif. Intell.}, 234:26--48, 2016.

\bibitem{chrisman1992reinforcement}
Lonnie Chrisman.
\newblock Reinforcement learning with perceptual aliasing: The perceptual
  distinctions approach.
\newblock In {\em AAAI}, pages 183--188. Citeseer, 1992.

\bibitem{Czumaj_lll}
Artur Czumaj and Christian Scheideler.
\newblock Coloring non-uniform hypergraphs: a new algorithmic approach to the
  general {L}ov\'asz local lemma.
\newblock In {\em Proceedings of the {E}leventh {A}nnual {ACM}-{SIAM}
  {S}ymposium on {D}iscrete {A}lgorithms ({S}an {F}rancisco, {CA}, 2000)},
  pages 30--39, 2000.

\bibitem{LLL}
Paul Erd{\H{o}}s and L\'{a}szl\'{o} Lov{\'a}sz.
\newblock Problems and results on {$3$}-chromatic hypergraphs and some related
  questions.
\newblock In {\em Infinite and finite sets ({C}olloq., {K}eszthely, 1973;
  dedicated to {P}. {E}rd{\H o}s on his 60th birthday), {V}ol. {II}}, pages
  609--627. Colloq. Math. Soc. J\'anos Bolyai, Vol. 10. North-Holland,
  Amsterdam, 1975.

\bibitem{etessami2010complexity}
Kousha Etessami and Mihalis Yannakakis.
\newblock On the complexity of {N}ash equilibria and other fixed points.
\newblock {\em SIAM Journal on Computing}, 39(6):2531--2597, 2010.

\bibitem{CompetitiveMDP}
Jerzy Filar and Koos Vrieze.
\newblock {\em Competitive Markov Decision Processes}.
\newblock Springer-Verlag New York, Inc., New York, NY, USA, 1996.

\bibitem{SrinivasanPerm}
David~G. Harris and Aravind Srinivasan.
\newblock A constructive algorithm for the {L}ov{\'a}sz local lemma on
  permutations.
\newblock In {\em SODA}, pages 907--925. SIAM, 2014.

\bibitem{HV}
Nicholas J.~A. Harvey and Jan Vondr{\'{a}}k.
\newblock An algorithmic proof of the {L}ov{\'{a}}sz local lemma via resampling
  oracles.
\newblock In Venkatesan Guruswami, editor, {\em {IEEE} 56th Annual Symposium on
  Foundations of Computer Science, {FOCS} 2015, Berkeley, CA, USA, 17-20
  October, 2015}, pages 1327--1346. {IEEE} Computer Society, 2015.

\bibitem{Planning}
Leslie~Pack Kaelbling, Michael~L Littman, and Anthony~R Cassandra.
\newblock Planning and acting in partially observable stochastic domains.
\newblock {\em Artificial intelligence}, 101(1):99--134, 1998.

\bibitem{RL}
Leslie~Pack Kaelbling, Michael~L Littman, and Andrew~W Moore.
\newblock Reinforcement learning: A survey.
\newblock {\em Journal of artificial intelligence research}, 4:237--285, 1996.

\bibitem{szege_meet}
Kashyap Babu~Rao Kolipaka and Mario Szegedy.
\newblock Moser and {T}ardos meet {L}ov{\'a}sz.
\newblock In {\em STOC}, pages 235--244. ACM, 2011.

\bibitem{Commu}
Vladimir Kolmogorov.
\newblock Commutativity in the random walk formulation of the {L}ov{\'{a}}sz
  local lemma.
\newblock {\em CoRR}, abs/1506.08547, 2015.

\bibitem{littman1998computational}
Michael~L Littman, Judy Goldsmith, and Martin Mundhenk.
\newblock The computational complexity of probabilistic planning.
\newblock {\em Journal of Artificial Intelligence Research}, 9(1):1--36, 1998.

\bibitem{madani1999undecidability}
Omid Madani, Steve Hanks, and Anne Condon.
\newblock On the undecidability of probabilistic planning and infinite-horizon
  partially observable {M}arkov decision problems.
\newblock In {\em AAAI/IAAI}, pages 541--548, 1999.

\bibitem{mike_stoc}
Michael Molloy and Bruce Reed.
\newblock Further algorithmic aspects of the local lemma.
\newblock In {\em S{TOC} '98 ({D}allas, {TX})}, pages 524--529. ACM, New York,
  1999.

\bibitem{moser}
Robin~A. Moser.
\newblock A constructive proof of the {L}ov{\'{a}}sz local lemma.
\newblock In Michael Mitzenmacher, editor, {\em Proceedings of the 41st Annual
  {ACM} Symposium on Theory of Computing, {STOC} 2009, Bethesda, MD, USA, May
  31 - June 2, 2009}, pages 343--350. {ACM}, 2009.

\bibitem{MT}
Robin~A. Moser and G{\'a}bor Tardos.
\newblock A constructive proof of the general {L}ov\'asz local lemma.
\newblock {\em J. ACM}, 57(2):Art. 11, 15, 2010.

\bibitem{mundhenk2000complexity}
Martin Mundhenk, Judy Goldsmith, Christopher Lusena, and Eric Allender.
\newblock Complexity of finite-horizon {M}arkov decision process problems.
\newblock {\em Journal of the ACM (JACM)}, 47(4):681--720, 2000.

\bibitem{papafocus}
Christos~H. Papadimitriou.
\newblock On selecting a satisfying truth assignment (extended abstract).
\newblock In {\em 32nd Annual Symposium on Foundations of Computer Science, San
  Juan, Puerto Rico, 1-4 October 1991}, pages 163--169. IEEE Computer Society,
  1991.

\bibitem{PapaTsitsi}
Christos~H Papadimitriou and John~N Tsitsiklis.
\newblock The complexity of {M}arkov decision processes.
\newblock {\em Mathematics of operations research}, 12(3):441--450, 1987.

\bibitem{Puterman}
Martin~L Puterman.
\newblock {\em {M}arkov decision processes: discrete stochastic dynamic
  programming}.
\newblock John Wiley \& Sons, 2014.

\bibitem{walksat}
Bart Selman, Henry~A. Kautz, and Bram Cohen.
\newblock Local search strategies for satisfiability testing.
\newblock In David~S. Johnson and Michael~A. Trick, editors, {\em Cliques,
  Coloring, and Satisfiability, Proceedings of a {DIMACS} Workshop, New
  Brunswick, New Jersey, USA, October 11-13, 1993}, volume~26 of {\em {DIMACS}
  Series in Discrete Mathematics and Theoretical Computer Science}, pages
  521--532. {DIMACS/AMS}, 1993.

\bibitem{aravind_08}
Aravind Srinivasan.
\newblock Improved algorithmic versions of the {L}ov\'asz local lemma.
\newblock In Shang{-}Hua Teng, editor, {\em SODA}, pages 611--620. {SIAM},
  2008.

\bibitem{Nikos}
Nikos Vlassis, Michael~L. Littman, and David Barber.
\newblock On the computational complexity of stochastic controller optimization
  in {POMDP}s.
\newblock {\em ACM Transactions on Computation Theory (TOCT)}, 4(4):12, 2012.

\end{thebibliography}

\end{document}